\documentclass{statsoc}

\title[Partition Function and Applications]{Asymptotic Properties of the Partition Function and Applications in Tail Index Inference of Heavy-Tailed Data}

\author[Danijel Grahovac {\it et al.}]{Danijel Grahovac}
\address{Department of Mathematics, Josip Juraj Strossmayer University,
Osijek,
Croatia.}
\email{dgrahova@mathos.hr}
\author{Mofei Jia}
\address{Department of Economics and Management, University of Trento,
Trento,
Italy.}
\email{mofei.jia-1@unitn.it}
\author{Nikolai N. Leonenko}
\address{Cardiff School of Mathematics,  Cardiff University,
Cardiff,
UK.}
\email{LeonenkoN@cardiff.ac.uk}
\author[Danijel Grahovac {\it et al.}]{Emanuele Taufer}
\address{Department of Economics and Management, University of Trento,
Trento,
Italy.}
\email{emanuele.taufer@unitn.it}

\usepackage{amssymb}
\usepackage{amsmath}
\usepackage{graphicx}
\usepackage{subcaption}
\newtheorem{theorem}{Theorem}
\newtheorem{lemma}{Lemma}
\newtheorem{remark}{Remark}
\newtheorem{definition}{Definition}
\DeclareMathOperator*{\argmin}{arg\,min}
\DeclareMathOperator*{\plim}{plim}

\begin{document}

\begin{abstract}
The so-called partition function is a sample moment statistic based on blocks of data and it is often used in the context of multifractal processes.

It will be shown that its behaviour is strongly influenced by the tail of the distribution underlying the data either in i.i.d. and weakly dependent cases.

These results will be exploited to develop graphical and estimation methods for the tail index of a distribution. The performance of the tools proposed is analyzed and compared with other methods by means of simulations and examples.

%Partition function is statistic that resembles the sample moment of the blocks of data. In this paper we establish asymptotic properties of partition %function in the setting of weakly dependent heavy-tailed distributed samples. Based on these results we define statistic called scaling function and %establish its asymptotic properties. Scaling functions can be used to analyze tails of the distribution based on data sample, both as a graphical and %estimation method. We present examples to illustrate the methods proposed.
\end{abstract}


\section{Introduction}

The partition function, also called empirical structure function, is a special kind of moment statistic. It is well known in the context of multifractal processes were it is used to estimate the so-called scaling function to describes scaling of moments in a stochastic process. Scaling functions have characteristic concave shape if the process is multifractal and can therefore be used to detect the multifractal nature of the process (see \cite{frisch1980} and \cite{mandelbrotFC1997}).

In Section 2 we study the asymptotic properties of the partition function in the setting of weakly dependent heavy-tailed data. In particular, we are interested in the rate of growth of the partition function, measured as the power of the sample size $n$ needed to normalize a sequence of partition functions when $n\to \infty$. The limiting behaviour involves a nice interplay between two classical results from probability theory: the law of large numbers and a generalized central limit theorem.

In Section 3 we present an application of the established results in the analysis of tail index of heavy-tailed data. Heavy-tailed distributions are of considerable importance in modeling a wide range of phenomena in finance, geology, hydrology, physics, queuing theory, and telecommunications. Pioneering work was done in \cite{mandelbrot1963}, where stable distributions with index less than $2$ have been advocated for describing fluctuations of cotton prices. In the field of finance, distributions of logarithmic asset returns can often be fitted extremely well by Student's $t$-distribution (see \cite{HeydeLeonenko2005} and the references therein). There are different ways to define the class of heavy-tailed distributions. In this paper, we say that the distribution of some random variable $X$ is heavy-tailed with index $\alpha>0$ if it has a regularly varying tail with index $-\alpha$. This implies that
$$P(|X|>x)=\frac{L(x)}{x^{\alpha}},$$
where $L(t),t>0$ is a slowly varying function, that is, $L(tx)/L(x) \to 1$ as $|x| \to \infty$, for every $t>0$. In particular, this implies that $E|X|^q < \infty$ for $q<\alpha$ and $E|X|^q=\infty$ for $q>\alpha$, which is sometimes also used to define heavy tails. The parameter $\alpha$ is called the tail index and measures the "thickness" of the tails.

We propose a graphical method that can be used as an exploratory method in order to detect heavy tails in the data. The graphs are used to detect the range of the tail index of the underlying distribution. Also, we establish an estimation method for the unknown tail index. Tail index estimators are usually based on upper order statistics and their asymptotic properties. As an alternative, \cite{meerschaert1998} proposed an estimator based on the asymptotics of the sum. In a certain way, underlying idea of our method is also based on the asymptotic properties of the sum. Our approach is however more general and independent of the results of \cite{meerschaert1998}. As we will see, the blocking structure of the partition function enables extracting  more information about the tail index. Moreover, we go beyond the i.i.d. case and consider weakly dependent samples. In Section 4 we present several examples based on simulated and real world data that illustrate the performance of the methods proposed. Section \ref{proofs} contains the proofs.

\section{The partition function}
Suppose we have a sample $X_1,\dots,X_n$ coming from a strictly stationary stochastic process $X_t, t \in \mathbb{Z}_+$ (discrete time) or $X_t, t \in \mathbb{R}_+$ (continuous time) which has a heavy-tailed marginal distribution with tail index $\alpha$. The partition function can be defined as follows:
\begin{equation}\label{partitionfunction}
S_q(n,t)=\frac{1}{\lfloor n/t \rfloor} \sum_{i=1}^{\lfloor n/t \rfloor} \left| \sum_{j=1}^{\lfloor t \rfloor}  X_{\lfloor t \rfloor(i-1)+j} \right|^q,
\end{equation}
where $q>0$ and $1 \leq t \leq n$. In words, we partition the data into consecutive blocks of length $\lfloor t \rfloor$, we sum each block and take the power $q$ of the absolute value of the sum. Finally, we average over all $\lfloor n/t \rfloor$ blocks. Notice that for $t=1$ one gets the usual empirical $q$-th absolute moment.

The partition function can also be viewed as a natural estimator of the $q$-th absolute moment of the stationary increment process. Indeed, suppose $\{Y_t\}$ is a process with stationary increments and one tries to estimate $E|Y(t)|^q$, for fixed $t>0$ based on a discretely observed sample $Y_1,\dots,Y_n$. It is natural to consider
\begin{equation*}
\frac{1}{\lfloor n/t \rfloor} \sum_{i=1}^{\lfloor n/t \rfloor} \left| Y_{i \lfloor t \rfloor } - Y_{(i-1) \lfloor t \rfloor }\right|^q.
\end{equation*}
If we denote one step increments as $X_i=Y(i)-Y(i-1)$, then this is equal to \eqref{partitionfunction}. When $\{Y_t\}$ is a L\'evy process, $X_1,\dots,X_n$ are independent identically distributed.

In studying the asymptotic properties of $S_q(n,t)$ we go beyond the i.i.d. case and consider $X_t$, $t \geq 0$ to be a strictly stationary process which satisfies a strong mixing condition with an exponentially decaying rate. More precisely, for two sub-$\sigma$-algebras, $\mathcal{A} \subset \mathcal{F}$ and $\mathcal{B} \subset \mathcal{F}$ on the same complete probability space $(\Omega, \mathcal{F}, P)$ define
$$a(\mathcal{A}, \mathcal{B}) = \sup\limits_{A \in \mathcal{A}, B \in \mathcal{B}}{\left| P(A \cap B) - P(A)P(B) \right|}.$$
Now for a process, $X_t$, $t \geq 0$, consider $\mathcal{F}_{t} = \sigma\{X_{s}, s \leq t\}, \ \mathcal{F}^{t+\tau} = \sigma\{X_{s}, s \geq t+\tau\}$. We say that $\{X_t\}$ has a strong mixing property if $a(\tau)= \sup_{t \geq 0} a(\mathcal{F}_{t}, \mathcal{F}^{t+\tau}) \allowbreak \to 0$ as $\tau \to \infty$. Strong mixing is sometimes also called $\alpha$-mixing. If $a(\tau)=O(e^{-b \tau})$ for some $b>0$ we say that the strong mixing property has an exponentially decaying rate. We note that results like Theorem 1 could probably be proven under some other weak form of dependence, but this form will cover a large variety of examples.

Asymptotic properties of $S_q(n,t)$ have been considered before in the context of multifractality detection (\cite{Sly2005}, \cite{Heyde2009}; see also \cite{HeydeandSly2008}). Instead of keeping $t$ fixed, we take it to be of the form $t=n^s$ for some $s\in (0,1)$, which allows the blocks to grow as the sample size increases. It is clear that then $S_q(n,n^s)$ will diverge since $s>0$. We are interested in the rate of divergence of this statistic, i.e., we consider the limiting behaviour of $\ln S_q(n,n^s) / \ln n$. One can think of the limiting value as the smallest power of $n$ needed to normalize the partition function such that  it will converge to some random variable not identically equal to zero.

The next theorem summarizes the main results on the rate of growth. We additionally assume that the sample has zero expectation in case it is finite. For practical purposes, this is not a restriction as one can always demean the starting sequence. For the case $\alpha\leq1$ this is not necessary. A special case of this theorem has been proved in \cite{Sly2005} and cited in \cite{Heyde2009}. The proof of the theorem is given in Section \ref{proofs}.

\begin{theorem}\label{thm:main}
Suppose $X_i, i \in \mathbb{N}$ is a strictly stationary sequence that has a strong mixing property with an exponentially decaying rate and suppose that $X_i, i \in \mathbb{N}$ has a heavy-tailed marginal distribution with tail index $\alpha>0$. Suppose also that $EX_i=0$ when $\alpha>1$. Then for $q>0$ and every $s \in (0,1)$
\begin{equation}\label{Rqs}
\frac{\ln S_q(n,n^s)}{\ln n} \overset{P}{\to} R_{\alpha}(q,s) :=
\begin{cases}
\frac{sq}{\alpha}, &\text{if } q \leq \alpha \text{ and } \alpha \leq 2,\\
s+\frac{q}{\alpha}-1, &\text{if } q > \alpha \text{ and } \alpha \leq 2,\\
\frac{sq}{2}, &\text{if } q \leq \alpha \text{ and } \alpha > 2,\\
\max\left\{ s+\frac{q}{\alpha}-1, \frac{sq}{2} \right\}, &\text{if } q > \alpha \text{ and } \alpha > 2,
\end{cases}
\end{equation}
as $n\to \infty$, where $\overset{P}{\to}$ stands for convergence in probability.
\end{theorem}

To explain the effects of the theorem, consider the simple case in which $X_1,X_2,\dots$ is a zero mean (if $\alpha>1$) i.i.d. sequence which is in the domain of normal attraction of some $\alpha$-stable random variable. This means that $\sum_{i=1}^n X_i / n^{1/\alpha}$ converges in distribution to some $Y$ with $\alpha$-stable distribution. Suppose first that $q<\alpha$. Consider
$$\frac{S_q(n,n^s)}{n^{\frac{sq}{\alpha}}}=\frac{1}{\lfloor n^{1-s} \rfloor} \sum_{i=1}^{\lfloor n^{1-s} \rfloor} \left| \frac{\sum_{j=1}^{\lfloor n^s \rfloor} X_{\lfloor n^s \rfloor(i-1)+j} }{n^{\frac{s}{\alpha}}}\right|^q.$$
When $n \to \infty$, each of the internal sums converges in distribution to an independent copy of $Y$. Since $q < \alpha$, $E|Y|^q$ is finite, so the weak law of large numbers applies and shows that the average tends to some non-zero and finite limit. For the case $q>\alpha$, the weak law cannot be applied, so the rate of growth will be higher, i.e.,
$$\frac{S_q(n,n^s)}{n^{s+\frac{q}{\alpha}-1}}=\frac{ \sum_{i=1}^{\lfloor n^{1-s} \rfloor} \left| \frac{\sum_{j=1}^{\lfloor n^s \rfloor} X_{\lfloor n^s \rfloor(i-1)+j} }{n^{\frac{s}{\alpha}}}\right|^q}{n^{(1-s) \frac{q}{\alpha}}}.$$
Internal sums again converge to independent copies of $Y$. Since $|Y|^q$ has $-\alpha/q$ regularly varying tail, it will be in the domain of normal attraction of $(\alpha/q)$-stable distribution, so the limit will be some positive random variable.

For the case $\alpha>2$, the variance is finite so the central limit theorem holds. When $q<\alpha$ the rate of growth has an intuitive explanation by arguments similar to those just mentioned above. When $q>\alpha$, interesting things happen. Note that asymptotics of the partition function is influenced by two things: averaging and the weak law on the one side and distributional limit arguments on the other side. It will depend on $s$ which of the two influences prevails. For larger $s$, $s+q/\alpha-1 < sq/2$ and the rate will be as in the case $2<\alpha <q$, i.e.,
$$\frac{S_q(n,n^s)}{n^{\frac{sq}{2}}}=\frac{1}{\lfloor n^{1-s} \rfloor} \sum_{i=1}^{\lfloor n^{1-s} \rfloor} \left| \frac{\sum_{j=1}^{\lfloor n^s \rfloor} X_{\lfloor n^s \rfloor(i-1)+j} }{n^{\frac{s}{2}}}\right|^q.$$
Internal sums converge in distribution to normal, which has all the moments finite and the weak law applies. But for small $s$, the rate will be the same as that for the case $\alpha<2$. What happens is that in this case internal sums have a small number of terms, so convergence to normal is slow, much slower than the effect of averaging. This is the reason why the rate is greater than $sq/2$.

\begin{remark}\label{remark1}
Note that in general, the normalizing sequence for partial sums can be of the form $n^{1/\alpha} L(n)$ for some slowly varying function $L$. This does not affect the rate of growth. Indeed, if $Z_n/n^aL(n) \overset{d}{\to} Z$ for some non-negative sequence $Z_n$, then for every $\varepsilon>0$,
\begin{align*}
&P \left( \frac{\ln Z_n}{\ln n} < a - \varepsilon \right) = P \left( Z_n < n^{a - \varepsilon} \right) =\\
&\quad P \left( \frac{Z_n}{n^aL(n)} < \frac{1}{L(n) n^{\varepsilon}} \right) \leq P \left( \frac{Z_n}{n^aL(n)} < \frac{1}{2 n^{\varepsilon}} \right) \to 0,
\end{align*}
since for $n$ large enough $n^{-\varepsilon} < L(n) < n^{\varepsilon}$, i.e. $\ln L(n) / \ln n \to 0$. Similar results can be obtained for the upper bound. The converse also holds, i.e., if $\ln Z_n / \ln n \overset{P}{\to} a$, then for some function $M$ such that $\ln M(n) / \ln n \to 0$ and $Z_n/n^a M(n) \overset{d}{\to} Z$ where $Z$ is a random variable not identically equal to zero.
\end{remark}

\begin{remark}
A natural question arises from the previous discussion, whether it is possible to identify a normalizing sequence and a distributional limit of $S_q(n,n^s)$. In some special cases the limit can be easily deduced. Suppose $X_i, i \in \mathbb{N}$ is an i.i.d. sequence with $\alpha$-stable distribution. When $q \leq \alpha$, the rate of growth will be $sq/\alpha$. Dividing the partition function with $n^{\frac{sq}{\alpha}}$ and using the scaling property of stable distributions yields
$$\frac{S_q(n,n^s)}{n^{\frac{sq}{\alpha}}}=\frac{1}{\lfloor n^{1-s} \rfloor} \sum_{i=1}^{\lfloor n^{1-s} \rfloor} \left| \frac{\sum_{j=1}^{\lfloor n^s \rfloor} X_{\lfloor n^s \rfloor(i-1)+j} }{n^{\frac{s}{\alpha}}}\right|^q  \overset{d}{=}  \frac{1}{\lfloor n^{1-s} \rfloor} \sum_{i=1}^{\lfloor n^{1-s} \rfloor} \left| X_i \right|^q.$$
Since $q \leq \alpha$, $E|X_i|^q < \infty$, so by the weak law of large numbers implies
$$\frac{S_q(n,n^s)}{n^{\frac{sq}{\alpha}}} \overset{P}{\to} E|X_1|^q, \quad n \to \infty.$$
On the other hand, when $q>\alpha$ weak law cannot be applied and the rate of growth is $s+\frac{q}{\alpha}-1$. Normalizing the partition function gives
$$\frac{S_q(n,n^s)}{n^{s+\frac{q}{\alpha}-1}}=\frac{ \sum_{i=1}^{\lfloor n^{1-s} \rfloor} \left| \frac{\sum_{j=1}^{\lfloor n^s \rfloor} X_{\lfloor n^s \rfloor(i-1)+j} }{n^{\frac{s}{\alpha}}}\right|^q}{n^{(1-s) \frac{q}{\alpha}}} \overset{d}{=} \frac{ \sum_{i=1}^{\lfloor n^{1-s} \rfloor} \left| X_i \right|^q}{n^{(1-s) \frac{q}{\alpha}}}.$$
$|X_i|^q$ have $-\alpha/q$ regularly varying tail, it will be in the domain of normal attraction of $(\alpha/q)$-stable distribution. $\alpha/q<1$, so centering is not needed and by generalized central limit theorem it follows that
$$\frac{S_q(n,n^s)}{n^{s+\frac{q}{\alpha}-1}} \overset{d}{\to} Y, \quad n \to \infty,$$
with $Y$ having $\alpha/q$-stable distribution. We leave studying of different limit types to be part of future research.
\end{remark}

\section{Applications in heavy tail inference}
In this section we discuss possible applications of the results developed in Section 2. First we develop the notion and properties of scaling functions. We use this to propose a graphical method that can be used to investigate the nature of tails of underlying distribution. Also we define an estimator of the unknown tail index.

\subsection{Scaling function}
The definition of the scaling function is based on the notion of partition function. For fixed $q$, we simply define the scaling function as the slope of the simple linear regression (with intercept) of $\ln S_q(n,n^s) / \ln n$ on $s$. Using the well known formula for the slope of linear regression line, we can formally state:
\begin{definition}
For $q>0$, the (empirical) scaling function at point $q$ is defined as
\begin{equation}\label{tauhat}
\hat{\tau}_{N,n}(q) = \frac{\sum_{i=1}^{N-1} \frac{i}{N} \frac{\ln S_q(n,n^{\frac{i}{N}})}{\ln n} - \frac{1}{N-1} \sum_{i=1}^{N-1} \frac{i}{N} \sum_{j=1}^{N-1} \frac{\ln S_q(n,n^{\frac{i}{N}})}{\ln n} }{ \sum_{i=1}^{N-1} \left(\frac{i}{N}\right)^2 - \frac{1}{N-1} \left( \sum_{i=1}^{N-1} \frac{i}{N} \right)^2 },
\end{equation}
where $N \in \mathbb{N}$.
\end{definition}
The definition can be justified by the following arguments. Using the notation of Theorem \ref{thm:main}, by Remark \ref{remark1}, there exists a function $M$ such that $\ln M(n) / \ln n \to 0$ and
$$\varepsilon_n:=\frac{S_q(n,n^s)}{n^{R_{\alpha}(q,s)} M(n)} \overset{d}{\to} \varepsilon,$$
where $\varepsilon$ is a random variable not identically equal to zero. This can be rewritten as
\begin{equation*}
\ln S_q(n,n^s) = R_{\alpha}(q,s) \ln n + \ln M(n) + \ln \varepsilon_n,
\end{equation*}
i.e.,
\begin{equation*}\label{regression}
\frac{\ln S_q(n,n^s)}{\ln n} = R_{\alpha}(q,s) + \frac{ \ln M(n)}{\ln n} + \frac{\ln \varepsilon_n}{\ln n}.
\end{equation*}
The term $\ln \varepsilon_n / \ln n$ can be considered as an error term in the regression of $\ln S_q(n,n^s) / \ln n$ on $q$ and $s$ with model function $R_{\alpha}(q,s)$. One should count on the intercept in the model, since for smaller $n$ the term $\ln M(n)/ \ln n$ may be non negligible. The possible nonzero mean of an error can be subtracted and considered as a part of the intercept. This complicated bivariate nonlinear regression can be simplified by noticing that the model function $R_{\alpha}(q,s)$ is approximately linear in $s$. Indeed, in the case $\alpha\leq 2$, the model function is exactly linear in $s$, while in case $\alpha>2$, this holds only approximately due to the maximum term when $q>\alpha$. Nonetheless, it makes sense to consider linear regression slope as the definition of the scaling function. $N$ in equation \eqref{tauhat} determines the number of data points used in regression.

Using Theorem \ref{thm:main} we can establish asymptotic properties of the scaling function.

\begin{theorem}\label{thm:tauhatasymptotic}
Suppose that the assumptions of the Theorem \ref{thm:main} hold. Then, for every $q>0$,
\begin{equation*}
\lim_{N \to \infty} \plim_{n \to \infty} \hat{\tau}_{N,n} (q) = \tau(q),
\end{equation*}
where $\plim$ stands for limit in probability and
\begin{equation}\label{tau}
\tau(q)=
\begin{cases}
\frac{q}{\alpha}, & \text{if } 0<q\leq \alpha \ \& \ \alpha\leq 2,\\
1, & \text{if } q>\alpha \ \& \ \alpha\leq 2,\\
\frac{q}{2}, & \text{if } 0<q\leq \alpha \ \& \ \alpha> 2,\\
\frac{q}{2}+\frac{2(\alpha-q)^2 (2\alpha+4q-3\alpha q)}{\alpha^3 (2-q)^2}, & \text{if } q>\alpha \ \& \ \alpha> 2.
\end{cases}
\end{equation}
\end{theorem}

Theorem shows that, asymptotically, shape of the scaling function calculated based on heavy-tailed data significantly depends on the value of the tail index $\alpha$. Plots of the asymptotic scaling function $\tau(q)$ for different values of $\alpha$ are shown on Figure \ref{tau plot}. When $\alpha \leq 2$, the scaling function is bilinear. In this case, first part of the plot is a line with slope $1/\alpha > 1/2$ and second part is a horizontal line with value $1$. A breakpoint occurs exactly at point $\alpha$. In case $\alpha>2$, $\tau(q)$ is approximately bilinear, the slope of the first part is $1/2$ and again the breakpoint is at $\alpha$. When $\alpha$ is large, i.e., $\alpha \to \infty$, it follows from \eqref{tau} that $\tau(q)\equiv q/2$. This case corresponds to data coming from a distribution with all moments finite, e.g., an independent normally distributed sample. This line will be referred to as the baseline. On figure \ref{tau plot} the baseline is shown by a dashed line. The case $\alpha\leq 2 ~(\alpha=0.5,1.0,1.5)$ and $\alpha>2~(\alpha=2.5,3.0,3.5,4.0)$ are shown by dot-dashed and solid lines, respectively.

\begin{figure}[H]
\caption{Plots of scaling function $\tau(q)$ against the moment $q$}
\label{tau plot}
\centering
\includegraphics[width=.55\textwidth]{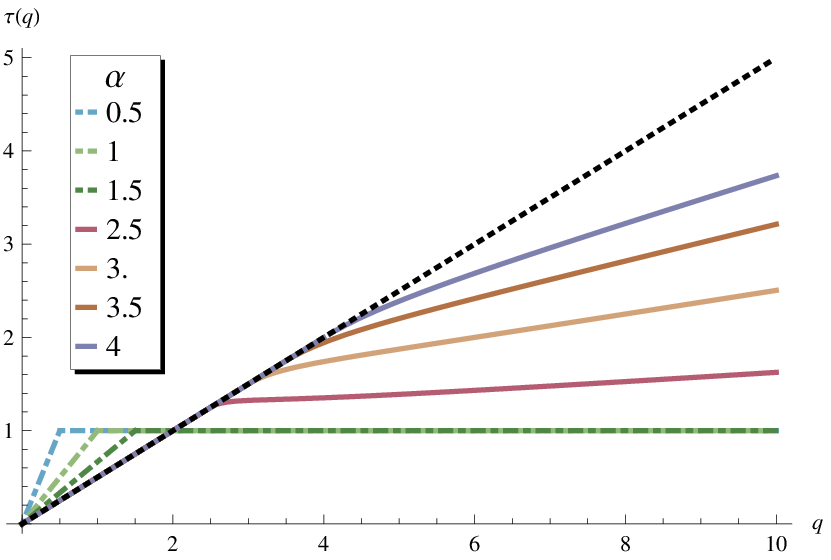}
\end{figure}

\subsection{Graphical method}

In this section we propose a graphical method useful for exploratory analysis of tails of the underlying distribution of the data. Since the scaling function shape is strongly influenced by the tail index value, this motivates the use of plots of the empirical scaling function to detect the tail index of a distribution. In particular the asymptotic results suggest that there should be sharp differences between the plots for distributions with infinite variance ($\alpha <2)$ and the others ($\alpha > 2$).

Based on a finite sample and chosen $N$, one can estimate the scaling function by equation \eqref{tauhat} for fixed value of $q$. Repeating this for a range of $q$ values makes it possible to give a plot of empirical scaling function $\hat{\tau}_{N,n}$.

By examining the plot and comparing it with the baseline it is possible to say something about the nature of tails of underlying distribution. If $\hat{\tau}_{N,n}(q)$ is above the baseline for $q<2$ and nearly horizontal afterward then true $\alpha$ is probably less than $2$. By examining the point where the graph breaks, one can roughly estimate the interval containing $\alpha$. If $\hat{\tau}_{N,n}(q)$ coincides with the baseline for $q<2$ and diverges from it after $q>2$, then true $\alpha$ is probably greater than $2$. The point at which deviation starts can be an estimate for $\alpha$. This also establishes a graphical method for distinguishing two cases, whether $\alpha\leq 2$ or $\alpha>2$.

If the graph coincides with the baseline, then we can suspect that the data does not exhibit heavy tails and that the moments are finite for the considered range of $q$. This way one can distinguish between heavy tails or not. We illustrate how the method works on simulated and real world examples in the next section.

\subsection{Estimation method}

Besides the graphical method, a simple estimation method for the unknown tail index can also be established based on asymptotic behaviour of the scaling function. Estimation of the tail index is a well known problem and there has been a range of estimators proposed. Probably the most popular estimator is the one proposed by \cite{hill1975}. Pickands and moment estimator are also heavily used. A nice survey of these estimators and their properties can be found in \cite{embrechts1997modelling} and \cite{dehaan2006}; see also \cite{meerschaert2003}. Various extensions and improvements of these estimators have been made making them suitable for data beyond i.i.d. case. As follows from the assumptions of Theorem \ref{thm:main}, the estimator defined here should work well for stationary strong mixing samples, thus extending the problem from the simplest i.i.d. case.

The basic idea of the method is to estimate $\alpha$ by fitting the empirical scaling function to the asymptotic form $\tau(q)$. This is done by the means of ordinary least squares. More precisely, for points $q_i \in (0,q_{max}), \ i=1,\dots,m$, calculate $\hat{\tau}_i=\hat{\tau}_{N,n}(q_i)$ using the Equation \eqref{tauhat}. Estimator is defined as
\begin{equation}\label{alphaMethod}
\hat{\alpha}=\argmin_{\alpha \in (0,\infty)} \sum_{i=1}^m (\hat{\tau}_i - \tau(q_i))^2.
\end{equation}
For practical reasons, due to the complexity of the expression for $\tau(q)$, method is divided into two cases: $\alpha \leq 2$ and $\alpha > 2$; i.e., corresponding part of $\tau(q)$ is used as a model function in \eqref{alphaMethod}, depending where the true value of $\alpha$ is. Therefore it is necessary to first detect whether we are in the case of infinite variance or not. This can be accomplished by using the graphical method described above. In the inconclusive case, it is advisable to compute both estimates and compare the quality of the fit.

\subsection{Plots of the empirical scaling functions}
The shape of the empirical scaling function is not always ideal as its asymptotic form. However, most plots are very close to their theoretical form. To illustrate this, we simulate $10$ independent samples of size $1000$ in six different settings. The first three cases studied are i.i.d. samples and others are stationary and weakly dependent, in accordance to the assumptions in Theorem \ref{thm:main}. Figure \ref{Ex1} summarizes the plots of the empirical scaling functions (dotted) together with the corresponding asymptotic form (solid) and the baseline (dot-dashed).

The first group of samples is generated from a stable distribution with stable index equal to $1$. The second one is generated from a Student $t$-distribution with $3$ degrees of freedom, a parameter that corresponds to the tail index. Recall that the probability density function of Student's $t$-distribution $T(\nu, \delta, \mu)$ is
\begin{equation}\label{pdf t}
\mathrm{student}[\nu, \delta, \mu](x)=\frac{\Gamma(\frac{\nu+1}{2})}{\delta\sqrt{\pi}\Gamma(\frac{\nu}{2})}\left(1+\left(\frac{x-\mu}{\delta}\right)^2\right)^{-\frac{\nu+1}{2}},~~~~~x\in \mathbb{R},
\end{equation}
where $\delta>0$ is the scaling parameter, $\nu$ the tail parameter (usually called degrees of freedom) and $\mu\in \mathbb{R}$ the location parameter. Figures \ref{Ex1a} and \ref{Ex1b} show that for both stable and Student case the empirical scaling functions are close to their theoretical form. Both plots are approximately bilinear and by identifying the breakpoint one can roughly guess the tail index value. Also, it is clear from shapes of the empirical scaling functions that variance is infinite in the first case and finite in the second. The third sample is generated from a standard normal distribution. From Figure \ref{Ex1c} one can surely doubt the existence of heavy-tails in these samples since the empirical scaling functions almost coincide with the baseline $q/2$. This shows that estimated scaling function have potential of providing self contained characterization of the tail.

Examples shown on Figures \ref{Ex1d}-\ref{Ex1f} are based on dependent data. Dependent samples are generated as sample paths of two types of stochastic processes: Ornstein-Uhlenbeck (OU) type processes and diffusions. Recall that a stochastic process $X=\{X_t, t\ge0 \}$ is said to be of OU type if it satisfies a SDE of the form
\begin{equation}\label{SDE ou}
\mathrm{d}X_t=-\lambda{X_t}\mathrm{d}t+\mathrm{d}L_{\lambda t}, \quad t\geq 0,
\end{equation}
where $L=\{L_t, t\ge0 \}$ is the background driving L\'evy process (BDLP) and $\lambda>0$. We consider strictly stationary solutions of SDE \eqref{SDE ou}. The $\alpha$-stable OU type process with parameter $\lambda>0$ and $0<\alpha<2$ introduced by \cite{Doob1942} is the solution of the SDE \eqref{SDE ou}, with $L=\{L_t, t \geq 0\}$ the standard $\alpha$-stable L\'evy motion (\cite{JanickiandWeron1994}). Since the distribution of increments for the BDLP $L$ is known in this case, we consider the Euler's scheme of simulation by replacing differentials in Equation \eqref{SDE ou} with differences. Student OU type process has been introduced in \cite{HeydeLeonenko2005}. It can be shown that there exists a strictly stationary stochastic process $X=\{X_t, t \geq 0\}$, which has a marginal $t$-distribution $T(\nu,\delta,\mu)$ with density function \eqref{pdf t} and BDLP $L$ such that \eqref{SDE ou} holds for arbitrary $\lambda>0$. This stationary process $X$ is referred to as the Student OU type process. Moreover, the cumulant transform of BDLP $L$ can be expressed as
\begin{equation*}\label{cumulant tran}
\kappa_{L_1}(\zeta)=\log{E\{e^{i\zeta{L_1}}\}}=i\zeta{\mu}-\delta|\zeta|\frac{K_{\nu/2-1}(\zeta{\mu})}{K_{\nu/2}(\zeta{\mu})}, ~~~~\zeta\in \mathbb{R},~~\zeta\ne0,
\end{equation*}
where $K$ is the modified Bessel function of the third kind and $\kappa_{L_1}(0)=0$ (\cite{HeydeLeonenko2005}). Since, for the Student OU process, the exact law of the increments of the BDLP is unknown, we use the approach introduced by \cite{LeonenkoandTaufer2009} to simulate discrete Student OU processes. This approach circumvents the problem of simulating the jumps of the BDLP and is easily applicable when an explicit expression of the cumulant transform is available; for more details see \cite{LeonenkoandTaufer2009}. Both OU processes considered can be shown to posses strong mixing property with an exponentially decaying rate (see \cite{Masuda2004}).
Last process considered is stationary Student diffusion. In order to define the Student diffusion, we introduce the stochastic differential equation (SDE):
\begin{equation}\label{sde student diff}
\mathrm{d}X_t=-\theta\left(X_t-\mu\right)\mathrm{d}t+\sqrt{\frac{2\theta\delta^2}{\nu-1}\left(1+\left(\frac{X_t-\mu}{\delta}\right)^2\right)}\mathrm{d}B_t, \quad t\geq 0,
\end{equation}
see \cite{Bibbyetal2005} and \cite{HeydeLeonenko2005}, where $\nu>1,~\delta>0, \mu\in \mathbb{R},~\theta>0$, and $B=\{B_t, t\ge0\}$ is a standard Brownian motion. SDE \eqref{sde student diff} admits a unique ergodic Markovian weak solution $X=\{X_t, t\ge0\}$ which is a diffusion process with the invariant symmetric scaled Student's $t$-distribution with probability density function \eqref{pdf t}. The diffusion process which solves the SDE \eqref{sde student diff} is called the Student diffusion. If $X_0\sim{T(\nu,\delta,\mu)}$, the Student diffusion is strictly stationary. According to \cite{LeonenkoandSuvak2010}, the Student diffusion is a strong mixing process with an exponentially decaying rate. For the simulation of paths of the Student diffusion process $X=\{X_t,t\ge0\}$ with known values of parameters, we have used the Milstein scheme (for details see \cite{Iacus2008}). Both OU processes were generated with auto-regression parameter $\lambda=1$ and diffusion was generated with $\theta=2$.

From examples on dependent data we can conclude that the shape of the scaling function is not affected with this weak form of dependence present. Again, empirical scaling functions are very near their asymptotic form.

We also report the mean of tail index values estimated by Equation \eqref{alphaMethod}: i.i.d. stable(1) - $1.29$, i.i.d. Student $t(3)$ - 3.30, i.i.d. $\mathcal{N}(0,1)$ - 4.73, stable(1) OU - 1.10, Student $t(3)$ OU - 3.34, Student $t(3)$ diffusion - 3.48.

\begin{figure}[H]
\centering
\caption{Plots of empirical scaling functions}
\label{Ex1}
\begin{subfigure}[b]{0.32\textwidth}
\centering
\includegraphics[width=\textwidth]{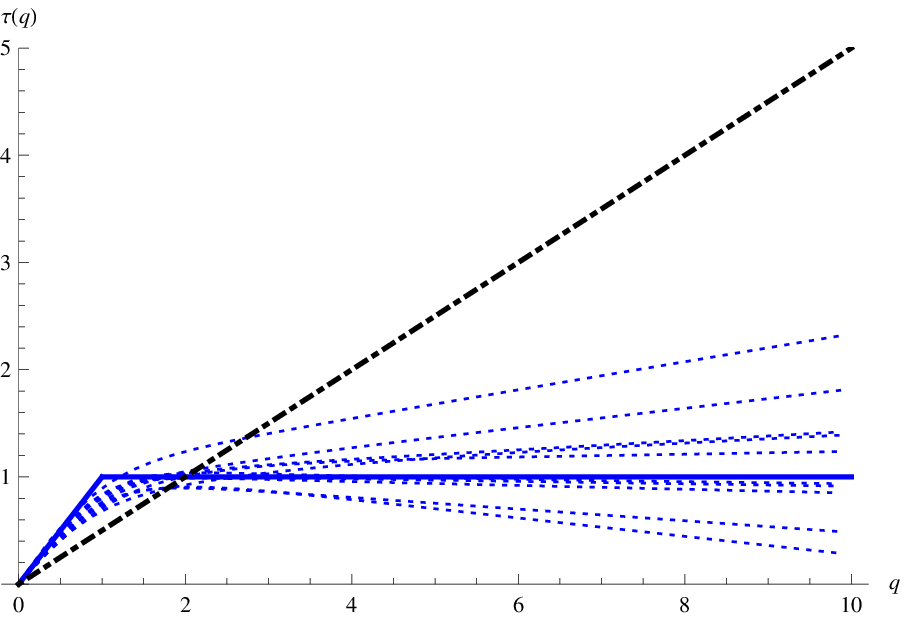}
\caption{Stable($1$) i.i.d.}
\label{Ex1a}
\end{subfigure}%
%\quad
\begin{subfigure}[b]{0.32\textwidth}
\centering
\includegraphics[width=\textwidth]{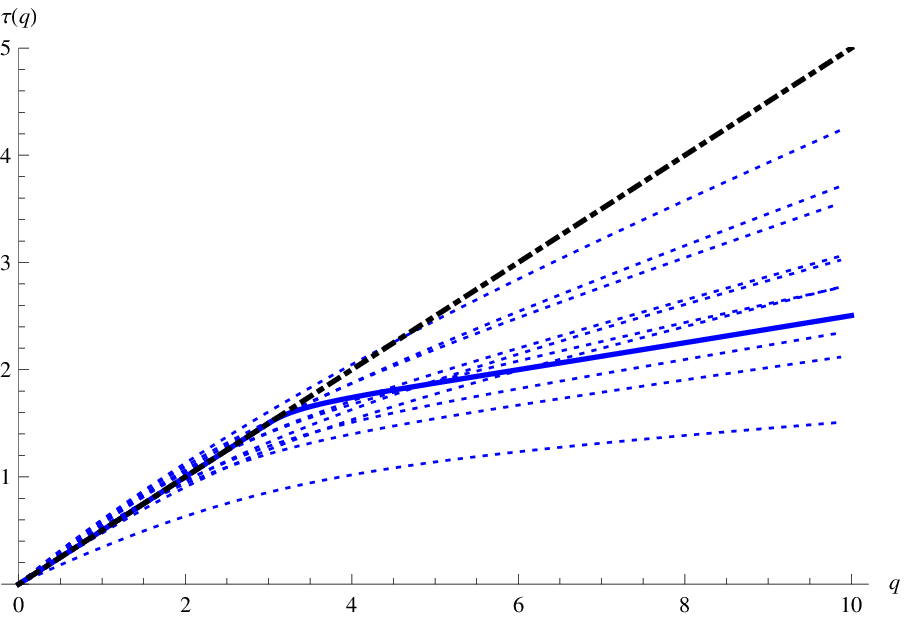}
\caption{Student $t(3)$  i.i.d.}
\label{Ex1b}
\end{subfigure}
%\quad
\begin{subfigure}[b]{0.32\textwidth}
\centering
\includegraphics[width=\textwidth]{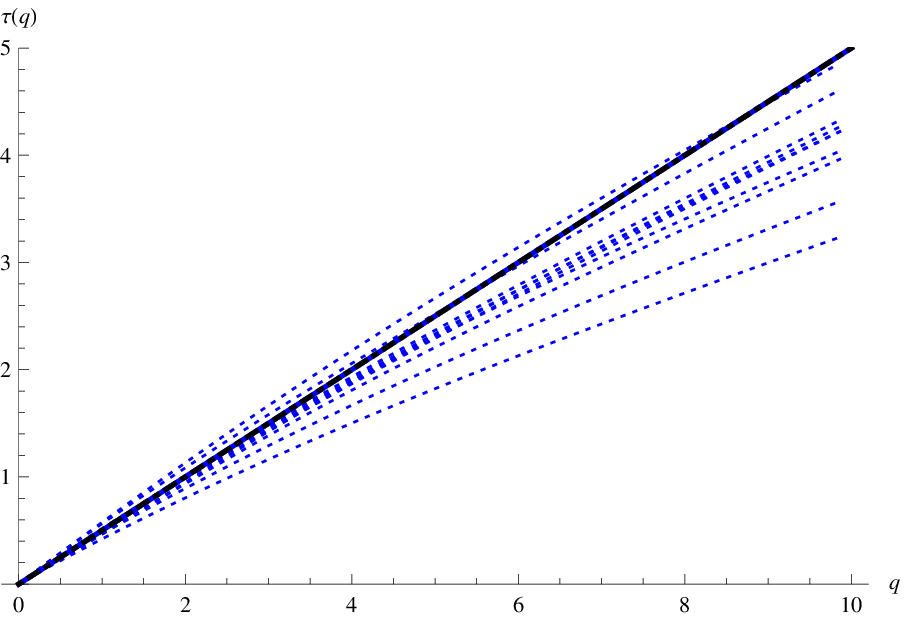}
\caption{$\mathcal{N}(0,1)$  i.i.d.}
\label{Ex1c}
\end{subfigure}
%\quad
\begin{subfigure}[b]{0.32\textwidth}
\centering
\includegraphics[width=\textwidth]{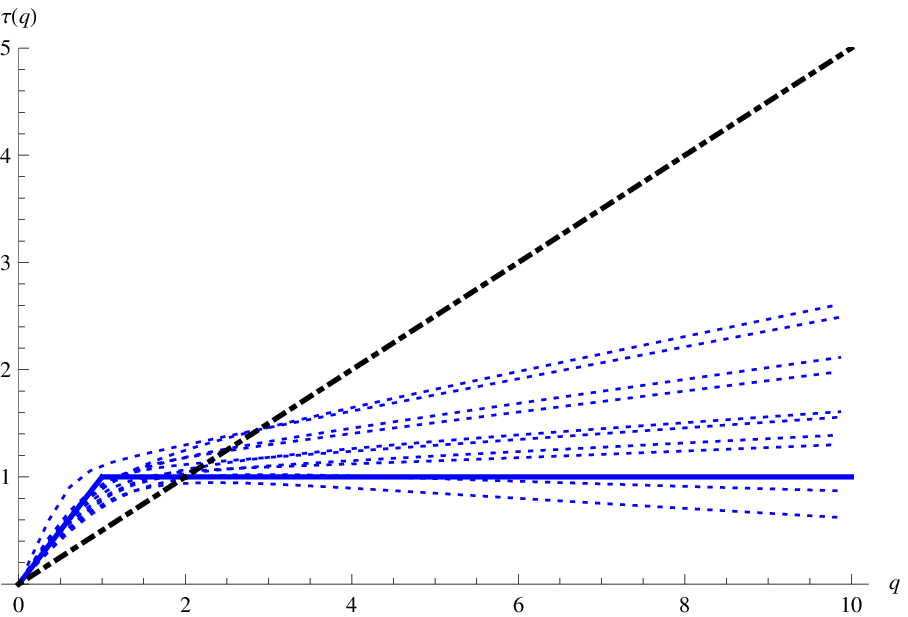}
\caption{Stable$(1)$ OU}
\label{Ex1d}
\end{subfigure}
%\quad
\begin{subfigure}[b]{0.32\textwidth}
\centering
\includegraphics[width=\textwidth]{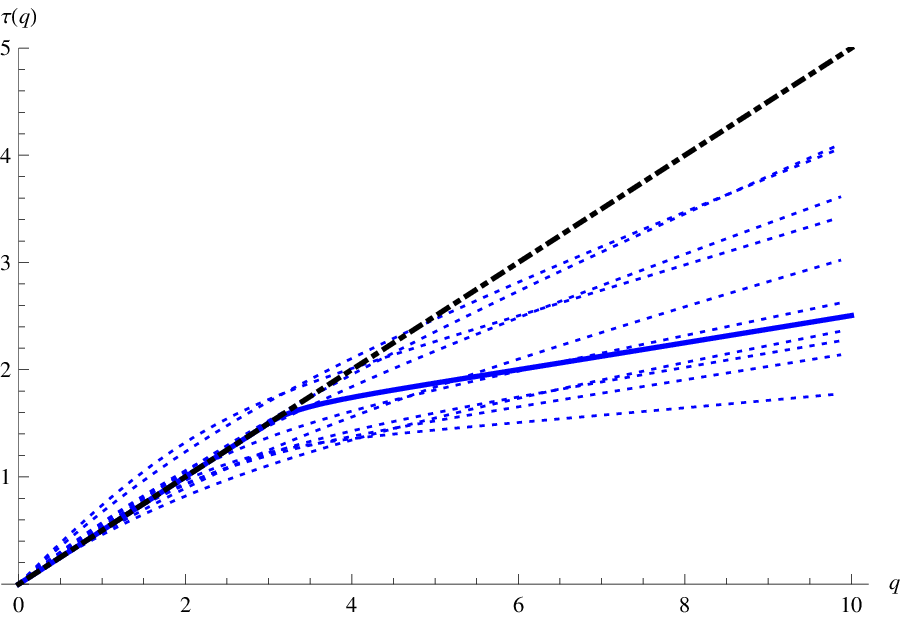}
\caption{Student $t(3)$ OU}
\label{Ex1e}
\end{subfigure}
%\quad
\begin{subfigure}[b]{0.32\textwidth}
\centering
\includegraphics[width=\textwidth]{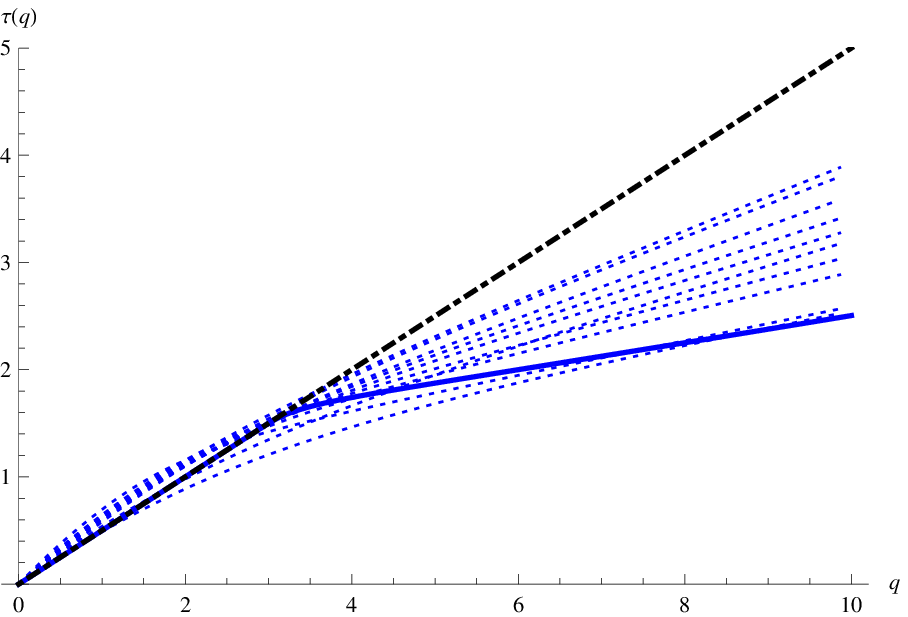}
\caption{Student $t(3)$ diffusion}
\label{Ex1f}
\end{subfigure}
\end{figure}

\section{Examples}
In this section we provide several examples to illustrate how the proposed methods work. We also compare the methods with existing ones.

Hill's estimator is the best known estimator of the tail index. For what follows, let $X_{(1)}\geq X_{(2)}\geq \cdots\geq X_{(n)}$ denote the order statistics of the sample $X_1,X_2,\cdots\,X_n$. For $1\leq k \leq n$, the Hill estimator based on $k$ upper order statistics is
\begin{equation}\label{Hillestimator}
\hat{\alpha}_{k}^{Hill}={\left(\frac{1}{k}\sum_{i=1}^{k}{\log{\frac{X_{(i)}}{X_{(k+1)}}}}\right)}^{-1}.
\end{equation}
Hill's estimator posses many desirable asymptotic properties, provided $k=k(n)$ is a sequence satisfying $\lim_{n\to\infty}{k(n)}=\infty$, and $\lim_{n\to\infty}{(k(n)/n)}=0$ (see, e.g. \cite{embrechts1997modelling}).

Another estimator of the tail index is the so called moment estimator proposed by \cite{dekkerseinmahldehaan1989}. Define for $r=1,2$
$$H_k^{(r)}=\frac{1}{k}\sum_{i=1}^{k} \left( \log{\frac{X_{(i)}}{X_{(k+1)}}} \right)^r.$$
A moment estimator based on $k$ order statistics is given by
\begin{equation}\label{Momentestimator}
\hat{\alpha}_{k}^{M}={\left( 1 + H_k^{(1)} + \frac{1}{2} \left( \frac{(H_k^{(2)})^2}{H_k^{(1)}} - 1 \right)^{-1}   \right)}^{-1}.
\end{equation}

Both equations actually yield a sequence of estimated values for different values of $k$. There exist several methods on how to choose $k$ in order to minimize asymptotic mean squared error. See \cite{beirlant2006} for a survey of these methods. Another possibility is to plot estimated values against $k$. Heuristic rule is to look for the place where the graph stabilizes and report this as the estimated value. For the Hill estimator this is usually called the Hill plot. We use this approach in the following examples to analyze the behaviour of the tail.

It is important to stress out that Hill plots cannot be used as graphical tool for establishing heavy tail property of the data as this can be misleading in cases when the tails are light. On the other hand, plots of the moment estimator can be used for this purpose (see \cite{resnick1997}). There are other exploratory tools for inspecting whether tails are heavy or not. One of the most frequently used tools is a variation of the QQ plot. By choosing $1\leq k \leq n$ one can plot the points
$$\left( - \ln \left( \frac{i}{k+1} \right) , \ln X_{(i)} \right), \quad i=1,\dots,k.$$
If the data is heavy-tailed with index $\alpha$ the plot should be roughly linear with slope $1/\alpha$. This is no more than the standard QQ plot of log-transformed data on exponential quantiles. This graphical method can be used to define estimator (\cite{resnick1996}), however, we will use it only as a exploratory tool. For $k=n$ this plot is sometimes called Zipf's plot. We will refer to it simply as the QQ plot.

\subsection{Example 1 - non heavy-tailed data}
With this example we try to illustrate the potential of scaling functions as a graphical method that can distinguish between heavy-tailed and light-tailed scenario. Different methods are tested on a random sample from standard normal distribution of size $2000$. Results are shown on Figure \ref{Ex1}. QQ plot for $500$ largest data points exhibits nonlinearity thus indicating that Pareto type tail is not a good fit for the data (Figure \ref{Ex1a}). For the purpose of tail detection we plot $1/\alpha$ values for the moment estimator. This is a more general extreme value index (see e.g. \cite{embrechts1997modelling} for details). From Figure \ref{Ex1b} one can see that plot stabilizes at a negative value near zero. This is indication of light tails. The scaling function shown on Figure \ref{Ex1c} is completely in accordance with this analysis. Indeed, it almost coincides with the baseline that corresponds to non heavy-tailed data.

\begin{figure}[H]
\centering
\caption{Example 1 - non heavy-tailed data}
\label{Ex1}
\begin{subfigure}[b]{0.32\textwidth}
\centering
\includegraphics[width=\textwidth]{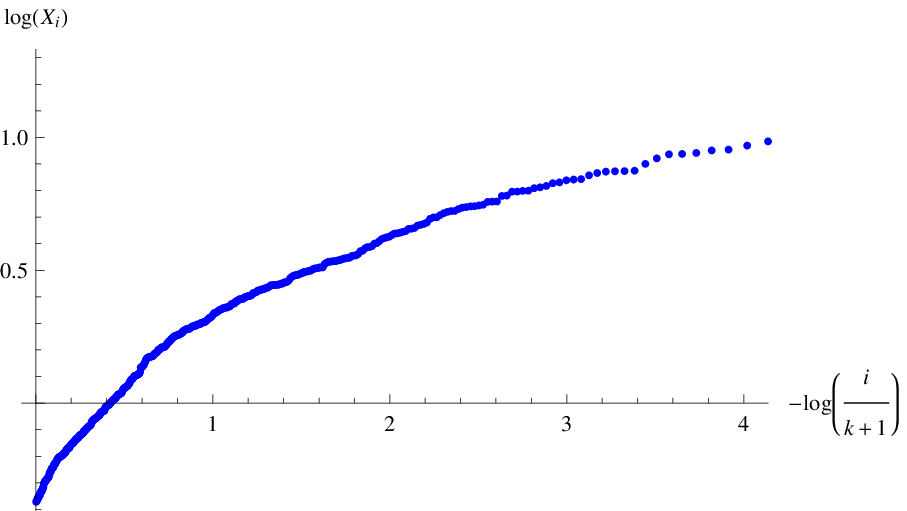}
\caption{QQ plot ($k=500$)}
\label{Ex1a}
\end{subfigure}%
%\quad
\begin{subfigure}[b]{0.32\textwidth}
\centering
\includegraphics[width=\textwidth]{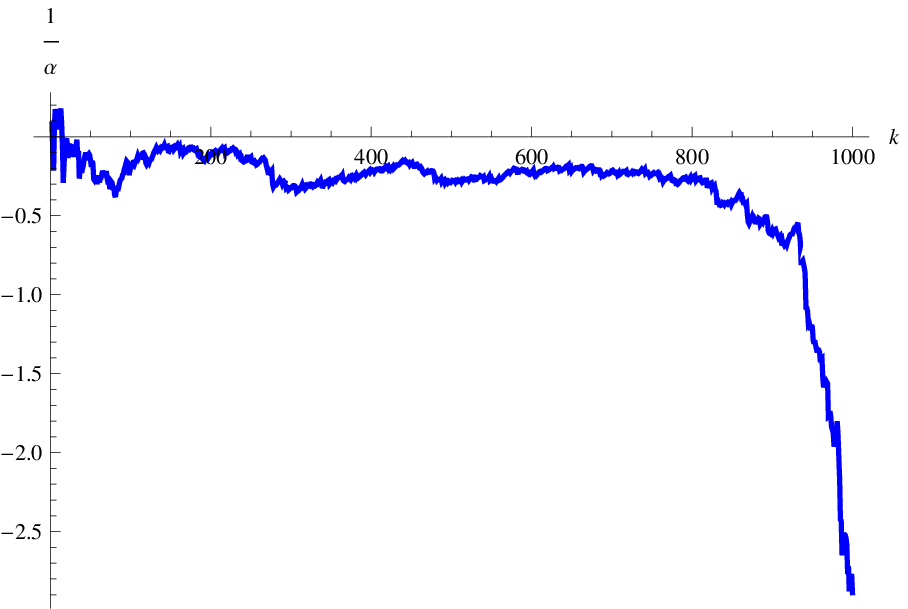}
\caption{Moment estimator plot}
\label{Ex1b}
\end{subfigure}
%\quad
\begin{subfigure}[b]{0.32\textwidth}
\centering
\includegraphics[width=\textwidth]{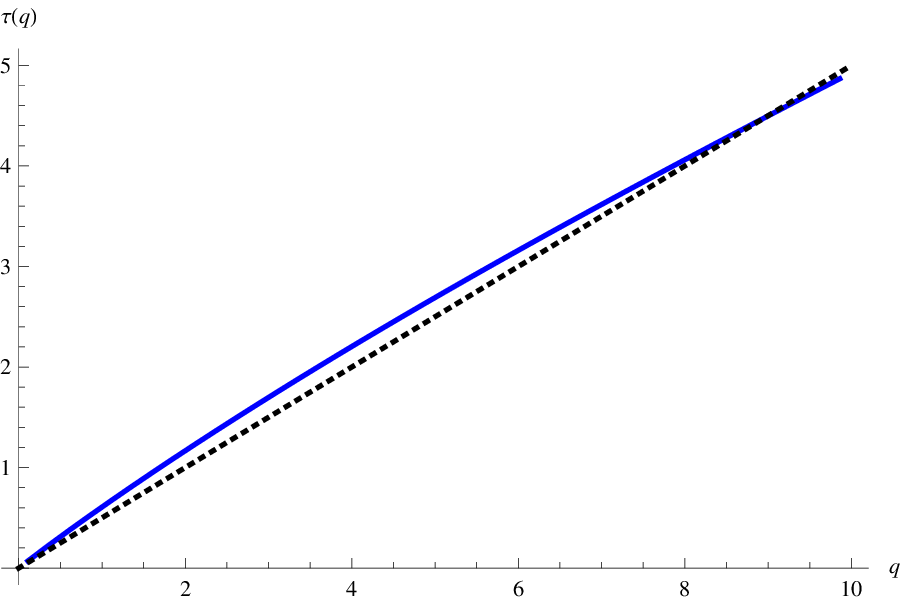}
\caption{Scaling function}
\label{Ex1c}
\end{subfigure}
\end{figure}

\subsection{Example 2 - departure from Pareto tail}
Hill's estimator, as well as many others, is known to behave poorly if the slowly varying function in the tail is far away from constant. We compare this behaviour with the performance of the scaling function estimator. Consider two distribution $F_1,F_2$ defined by their survival functions
\begin{align}
\overline{F}_1(x)=1-F_1(x)=\frac{1}{x^{\frac{1}{2}}}, \quad x \geq 1,\\
\overline{F}_2(x)=1-F_2(x)=\frac{e^{\frac{1}{2}}}{x^{\frac{1}{2}} \ln x}, \quad x \geq e.
\end{align}
Both distributions are heavy-tailed with tail index equal to $1/2$. We generate samples from these two distributions with $5000$ observations. The corresponding Hill plots are shown in Figure \ref{Ex2a}. While for the Pareto distribution $F_1$ Hill provides very good results, for $F_2$ it is impossible to draw any conclusion about the value of the tail index. The plot fails to stabilize at some value and produces a departure from the true index value. This is sometimes called Hill horror plot (see \cite{embrechts1997modelling}). The result is similar with the moment estimator: a non-constant slowly varying function in the tail produces a significant bias, as shown on Figure \ref{Ex2b}.

Figure \ref{Ex2c} shows the empirical scaling functions for the same samples together with the theoretical one and the baseline. One can see that scaling functions are very close to the theoretical one, especially in the first part of the plot, before the breakpoint. It seems that non-constant slowly varying function affects the estimation but the effect is not so dramatic as for the other two estimator. Calculating estimates using \eqref{alphaMethod} yields values $\hat{\alpha}_1=0.53$ and $\hat{\alpha}_2=0.67$.

\begin{figure}[H]
\centering
\caption{Example 2 - departure from Pareto tail}
\label{Ex2}
\begin{subfigure}[b]{0.32\textwidth}
\centering
\includegraphics[width=\textwidth]{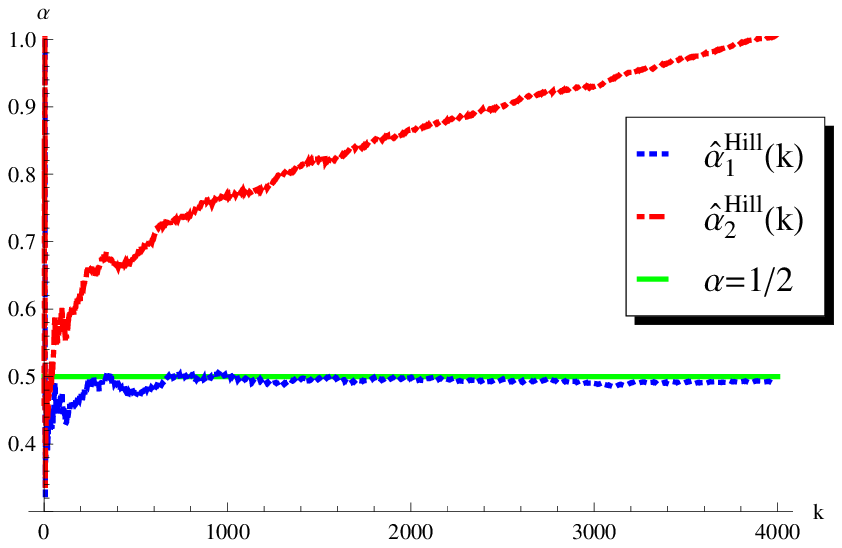}
\caption{Hill plot}
\label{Ex2a}
\end{subfigure}%
%\quad
\begin{subfigure}[b]{0.32\textwidth}
\centering
\includegraphics[width=\textwidth]{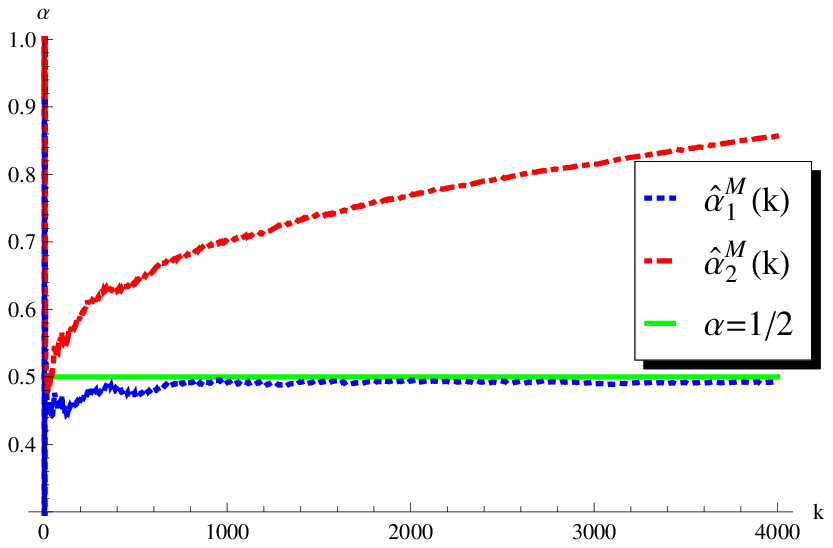}
\caption{Moment estimator plot}
\label{Ex2b}
\end{subfigure}
%\quad
\begin{subfigure}[b]{0.32\textwidth}
\centering
\includegraphics[width=\textwidth]{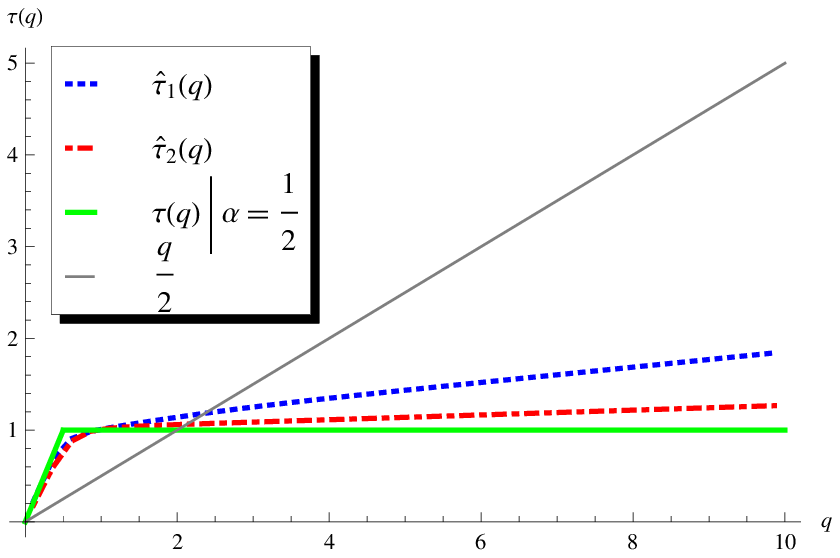}
\caption{Scaling function}
\label{Ex2c}
\end{subfigure}
\end{figure}

\subsection{Example 3 - Danish fire insurance claims}
Finally, we present a practical example. This example is similar to Example 6.2.9 from \cite{embrechts1997modelling}. The data corresponds to Danish fire insurance claims in the period from 1980 to 1990. There are $2167$ observations and the amounts are in millions of Danish Kroner.\footnote{The data can be obtained from: http://www.ma.hw.ac.uk/~mcneil/data.html} The analysis made in \cite{embrechts1997modelling} suggests a tail index estimate around $1.618$ (see Example 6.4.5). Hill and moment estimator plots (figures \ref{Ex3a} and \ref{Ex3b}) confirm the index value is around $1.5$. Empirical scaling function calculated from the data together with the baseline is shown on Figure \ref{Ex3c}. The data has been demeaned to adjust to the assumptions of the Theorem \ref{thm:main}. The scaling function is approximately bilinear: the first part of the plot has slope greater than the baseline and second part is nearly horizontal. This points out that the variance is infinite. The breakpoint occurs at around $1.5$, which indicates the possible value of the tail index. Estimating by Equation \eqref{alphaMethod} yields the value $1.419$, consistent with the previous analysis done in \cite{embrechts1997modelling}.

\begin{figure}[H]
\centering
\caption{Example 3 - Danish fire insurance claims}
\label{Ex3}
\begin{subfigure}[b]{0.32\textwidth}
\centering
\includegraphics[width=\textwidth]{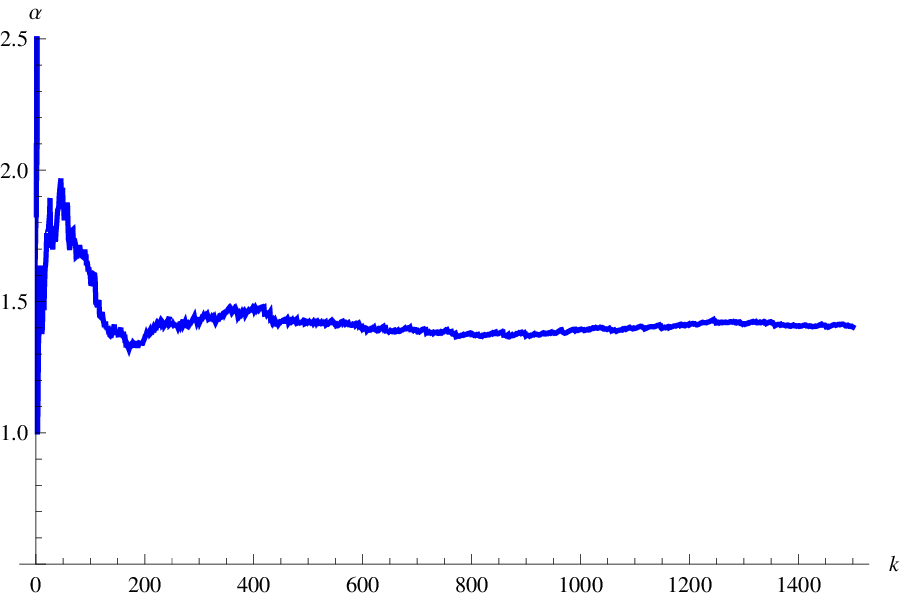}
\caption{Hill plot}
\label{Ex3a}
\end{subfigure}%
%\quad
\begin{subfigure}[b]{0.32\textwidth}
\centering
\includegraphics[width=\textwidth]{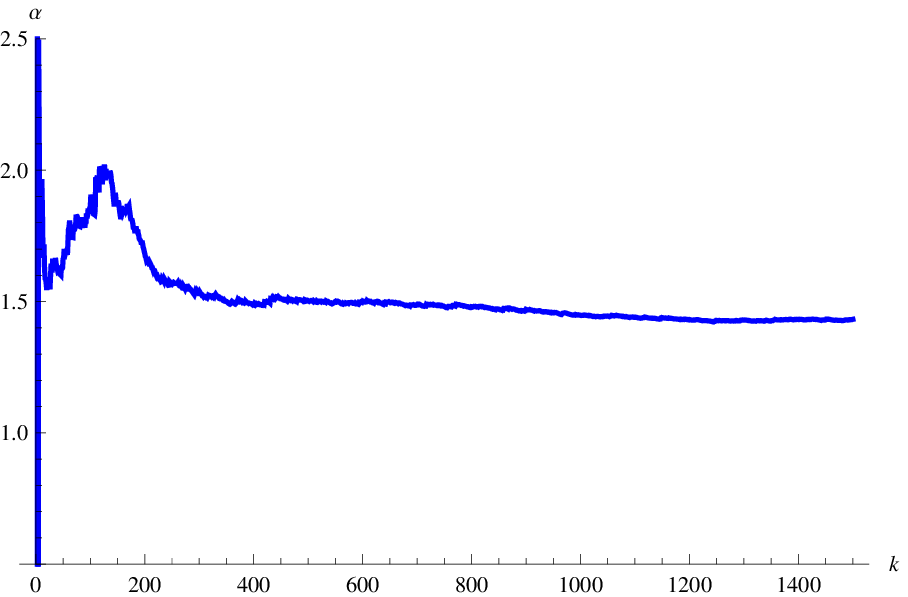}
\caption{Moment estimator plot}
\label{Ex3b}
\end{subfigure}
%\quad
\begin{subfigure}[b]{0.32\textwidth}
\centering
\includegraphics[width=\textwidth]{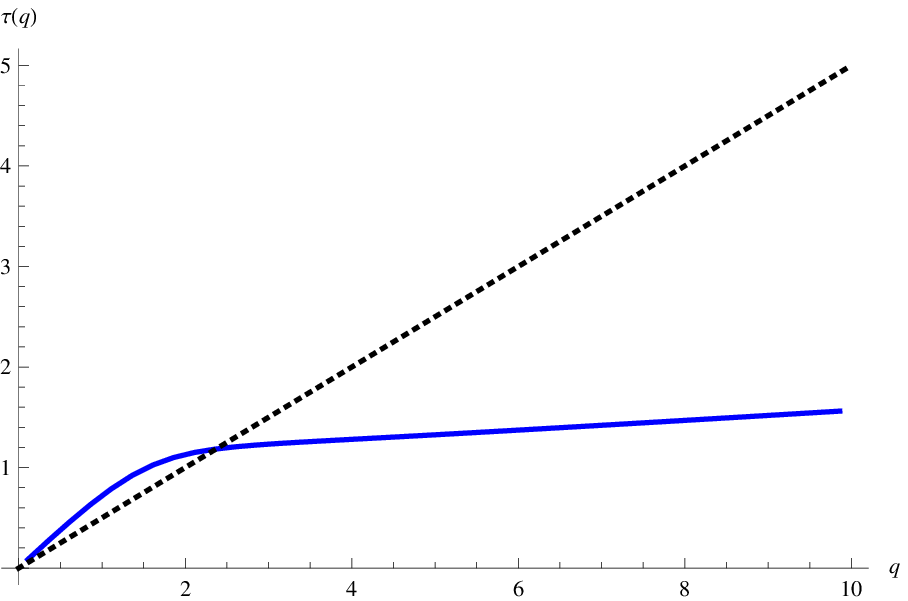}
\caption{Scaling function}
\label{Ex3c}
\end{subfigure}
\end{figure}

\section{Summary and discussion}
In the first part of the paper we present results on the asymptotic behaviour of a special kind of moment statistic. This behaviour is strongly influenced by the existence of the moments of the underlying distribution. The limiting behaviour involves a nice interplay between two classical results from probability theory: the law of large numbers and a generalized central limit theorem. Established results provide a deeper insight into the rate of divergence of sample moments. From a probabilistic point of view, the result is interesting in its own matter.

In the second part we discuss possible applications in the context of  tail index estimation. We establish a quantity called the scaling function which can be though of as the signature of heavy tails. It has the ability to reflect tail properties on a single plot. This property is used for establishing the procedures for investigating tail behaviour and estimating tail index.

\section{Proofs}\label{proofs}
The following is a version of Rosenthal's inequality for strong mixing sequences needed in the proof, precisely Theorem 2 in section 1.4.1 of \cite{doukhan1994}:

\begin{lemma}\label{lemma:Rosenthal}
Fix $q>0$ and suppose $(Y_k)$ is a sequence of random variables and let $a_Y(m)$ be the corresponding strong mixing coefficient function. Suppose that there exists $\zeta>0$ and $c\geq q, c \in \mathbb{N}$ such that
\begin{equation*}
\sum_{m=1}^{\infty} (m+1)^{2c-2} \left(a_Y(m)\right)^{\frac{\zeta}{2c+\zeta}} < \infty,
\end{equation*}
and suppose $E|Y_k|^{q+\zeta}<\infty$ and $Y_k$ are centered for all $k$. Then there exists some constant $K$ depending only on $q$ and $a_Y(m)$ such that
\begin{equation*}
E\left| \sum_{k=1}^l Y_k \right|^q \leq K D(q,\zeta,l),
\end{equation*}
where
\begin{equation*}
D(q,\zeta,l)=
\begin{cases}
L(q,0,l), & \text{if } 0<q\leq 1,\\
L(q,\zeta,l), & \text{if } 1<q\leq 2,\\
\max \left\{ L(q,\zeta,l), \left( L(2,\zeta,l) \right)^{\frac{q}{2}} \right\}, & \text{if } q> 2,
\end{cases}
\end{equation*}
\begin{equation*}
L(q,\zeta,l)= \sum_{k=1}^l \left( E \left| Y_k \right|^{q+\zeta} \right)^{\frac{q}{q+\zeta}}.
\end{equation*}
\end{lemma}

\bigskip

\begin{proof}[Proof of theorem \ref{thm:main}]
We split the proof into three parts depending whether $q>\alpha$, $q<\alpha$ or $q=\alpha$.\\
(a) Let $q>\alpha$. First we show an upper bound for the limit in probability.\\
Let $\epsilon>0$. Notice that
$$n^{\frac{\ln S_q(n,n^s)}{\ln n}}=S_q(n,n^s)=\frac{1}{\lfloor n^{1-s} \rfloor} \sum_{i=1}^{\lfloor n^{1-s} \rfloor} \left| \sum_{j=1}^{\lfloor n^s \rfloor} X_{\lfloor n^s \rfloor(i-1)+j} \right|^q.$$
Let $\delta>0$ and define
\begin{align*}
Y_{j,n} &= X_j \mathbf{1}\left( |X_j|\leq n^{\frac{1}{\alpha}+\delta} \right), \quad j=1,\dots,n, \ n\in \mathbb{N},\\
Z_{j,n} &= Y_{j,n} - EY_{j,n},\\
\xi_i &= \left| \sum_{j=1}^{\lfloor n^s \rfloor} Z_{n^s(i-1)+j, n} \right|^q, \quad i=1,\dots, \lfloor n^{1-s} \rfloor.
\end{align*}
Before splitting the cases based on different $\alpha$ values, we derive some facts that will be used later. Due to stationarity, for fixed $n$, the $\xi_i$'s are identically distributed so that $E\left[ \frac{1}{k} \sum_{i=1}^k \xi_i \right] = E[\xi_1]$. $Y_{j,n}$ has finite moments of all orders and by using Karamata's theorem (\cite{embrechts1997modelling}), for arbitrary $r>\alpha$ it follows
\begin{equation}\label{momentbound}
\begin{aligned}
&E|Z_{j,n}|^r \leq E|Y_{j,n}|^r = \int_0^{\infty} P(|Y_{j,n}|^r>x)dx = \int_0^{n^{r(\frac{1}{\alpha} + \delta)}} P(|X_j|^r>x)dx\\
&= \int_0^{n^{r(\frac{1}{\alpha} + \delta)}} L(x^{\frac{1}{r}}) x^{-\frac{\alpha}{r}} dx \leq C_1 n^{r(\frac{1}{\alpha} + \delta)(-\frac{\alpha}{r}+1)} = C_1 n^{\frac{r}{\alpha}-1+\delta (r - \alpha)}
\end{aligned}
\end{equation}

Next, notice that, for fixed $n$, $Z_{j,n}, j=1,\dots,n$ is a stationary sequence. By definition $EZ_{j,n}=0$ and also $E|Z_{j,n}|^{q+\zeta}<\infty$ for every $\zeta>0$. Since $Z_{j,n}$ is no more than a measurable transformation of $X_j$, the mixing properties of $Z_{j,n}$ are inherited from those of sequence $X_j$. This means that there exists a constant $b>0$ such that the mixing sequence $a_X (m)=a_Z(m)=O(e^{-bm})$ as $m\to \infty$. It follows that
$$\sum_{m=1}^{\infty} (m+1)^{2c-2} \left(a_Z(m)\right)^{\frac{\zeta}{2c+\zeta}} \leq \sum_{m=1}^{\infty} (m+1)^{2c-2} K_1 e^{-bm\frac{\zeta}{2c+\zeta}}< \infty$$
for every choice of $c\in \mathbb{N}$ and $\zeta>0$. Hence we can apply Lemma \ref{lemma:Rosenthal} for $n$ fixed to get
\begin{equation}
E \xi_1 = E \left| \sum_{j=1}^{\lfloor n^s \rfloor} Z_{j,n} \right|^q  \leq
\begin{cases}
K L(q,0,\lfloor n^s \rfloor), & \text{if } 0<q\leq 1,\\
K L(q,\zeta,\lfloor n^s \rfloor), & \text{if } 1<q\leq 2,\\
K \max \left\{ L(q,\zeta,\lfloor n^s \rfloor), \left( L(2,\zeta,\lfloor n^s \rfloor) \right)^{\frac{q}{2}} \right\}, & \text{if } q> 2.
\end{cases}
\end{equation}
Notice that none of the previous arguments uses assumptions on $\alpha$. Now we split the cases:
\begin{itemize}
\item[\textbullet$\alpha>2$] Because for $q>\alpha$, utilizing Equation \eqref{momentbound}, we can choose $\zeta$ so small such that $\zeta<q \delta \alpha$ (in order to achieve $n^{-\frac{q}{q+\zeta}(1+\delta \alpha)} < n^{-1}$)  to obtain
\begin{align}
L(q,\zeta,\lfloor n^s \rfloor) &=  \sum_{j=1}^{\lfloor n^s \rfloor} \left( E \left| Z_{j,n} \right|^{q+\zeta} \right)^{\frac{q}{q+\zeta}} \leq C_2 n^s n^{\left(\frac{q+\zeta}{\alpha}-1+\delta (q+\zeta - \alpha)\right)\left( \frac{q}{q+\zeta} \right)} \nonumber \label{1casebound} \\
&= C_2 n^{s+\frac{q}{\alpha}-\frac{q}{q+\zeta}(1+\delta \alpha ) + \delta q} \leq C_2 n^{s+\frac{q}{\alpha}-1+\delta q}, \\
\left( L(2,\zeta,\lfloor n^s \rfloor) \right)^{\frac{q}{2}} &=  \left( \sum_{j=1}^{\lfloor n^s \rfloor} \left( E \left| Z_{j,n} \right|^{2+\zeta} \right)^{\frac{2}{2+\zeta}} \right)^{\frac{q}{2}} \nonumber \\
&\leq n^{\frac{sq}{2}} \left( E|X_1|^{2+\zeta}\right)^{\frac{q}{2+\zeta}} \leq C_3 n^{\frac{sq}{2}} \nonumber.
\end{align}
Hence $E \xi_1 \leq C_4 n^{\max\left\{ s+\frac{q}{\alpha}-1+\delta q, \frac{sq}{2} \right\}}$.

\item[\textbullet$1<\alpha \leq2$] Bound for  $L(q,\zeta,\lfloor n^s \rfloor)$ is the same as in \eqref{1casebound}, so if $\alpha<q\leq 2$ we have $E \xi_1 \leq K  L(q,\zeta,\lfloor n^s \rfloor) \leq C_4 n^{s+\frac{q}{\alpha}-1+\delta q}$. If $q>2$, using Equation \eqref{momentbound} and choosing $\zeta<2 \delta \alpha$, yields
\begin{align*}
\left( L(2,\zeta,\lfloor n^s \rfloor) \right)^{\frac{q}{2}} &=  \left( \sum_{j=1}^{\lfloor n^s \rfloor} \left( E \left| Z_{j,n} \right|^{2+\zeta} \right)^{\frac{2}{2+\zeta}} \right)^{\frac{q}{2}} \leq n^{\frac{sq}{2}} \left( C_1 n^{\frac{2+\zeta}{\alpha} -1 +\delta (2+\zeta - \alpha)} \right)^{\frac{q}{2+\zeta}}\\
&\leq C_5 n^{\frac{sq}{2}+\frac{q}{\alpha}-\frac{q}{2+\zeta}(1+\delta \alpha) +\delta q} \leq C_5 n^{\frac{sq}{2}+\frac{q}{\alpha}-\frac{q}{2}+\delta q}.
\end{align*}
But, for $q>2$
$$s+\frac{q}{\alpha}-1 - \frac{sq}{2} - \frac{q}{\alpha} + \frac{q}{2}=\left( 1-\frac{q}{2} \right) \left( s-1 \right)>0,$$
so, $s+\frac{q}{\alpha}-1 > \frac{sq}{2} + \frac{q}{\alpha} - \frac{q}{2}$ and
$$K \max \left\{ L(q,\zeta,\lfloor n^s \rfloor), \left( L(2,\zeta,\lfloor n^s \rfloor) \right)^{\frac{q}{2}} \right\} \leq C_6 n^{s+\frac{q}{\alpha}-1+\delta q}.$$
We conclude that for every $q>\alpha$, $E \xi_1 \leq C_7 n^{s+\frac{q}{\alpha}-1+\delta q}$.

\item[\textbullet $ 0<\alpha \leq 1$] If $q>1$ we can repeat arguments from the previous case. If $\alpha<q \leq 1$, again by \eqref{momentbound}
$$L(q,0,\lfloor n^s \rfloor)=\sum_{j=1}^{\lfloor n^s \rfloor} E|Z_{j,n}|^q \leq C_2 n^{s+\frac{q}{\alpha}-1+\delta (q-\alpha)} \leq C_2 n^{s+\frac{q}{\alpha}-1+\delta q} ,$$
so, for every $q>\alpha$, $E \xi_1 \leq C_8 n^{s+\frac{q}{\alpha}-1+\delta q}$.\\
\end{itemize}

Next, notice that
$$P \left( \max_{i=1,\dots,n} |X_i| > n^{\frac{1}{\alpha} + \delta} \right) \leq \sum_{i=1}^n P \left(|X_i| > n^{\frac{1}{\alpha} + \delta} \right)
\leq n \frac{L(n^{\frac{1}{\alpha}+\delta})}{(n^{\frac{1}{\alpha}+\delta})^{\alpha}}
\leq C_9 \frac{L(n^{\frac{1}{\alpha}+\delta})}{n^{\alpha \delta}},$$
\begin{align*}
S_q(n,n^s) &= \frac{1}{\lfloor n^{1-s} \rfloor} \sum_{i=1}^{\lfloor n^{1-s} \rfloor} \left| \sum_{j=1}^{\lfloor n^s \rfloor} \left( X_{\lfloor n^s \rfloor(i-1)+j} +  EY_{\lfloor n^s \rfloor(i-1)+j, n} - EY_{\lfloor n^s \rfloor(i-1)+j, n} \right) \right|^q\\
&\leq \max \{1,2^{q-1}\}  \frac{1}{\lfloor n^{1-s} \rfloor} \sum_{i=1}^{\lfloor n^{1-s} \rfloor} \left| \sum_{j=1}^{\lfloor n^s \rfloor} \left( X_{\lfloor n^s \rfloor(i-1)+j} -  EY_{\lfloor n^s \rfloor(i-1)+j, n} \right) \right|^q \\
&\qquad + \max \{1,2^{q-1}\} \frac{1}{\lfloor n^{1-s} \rfloor} \sum_{i=1}^{\lfloor n^{1-s} \rfloor} \left| \sum_{j=1}^{\lfloor n^s \rfloor} EY_{\lfloor n^s \rfloor(i-1)+j, n} \right|^q \\
&\leq \max \{1,2^{q-1}\}  \frac{1}{\lfloor n^{1-s} \rfloor} \sum_{i=1}^{\lfloor n^{1-s} \rfloor} \left| \sum_{j=1}^{\lfloor n^s \rfloor} \left( X_{\lfloor n^s \rfloor(i-1)+j} -  EY_{\lfloor n^s \rfloor(i-1)+j, n} \right) \right|^q  \\
&\qquad + \max \{1,2^{q-1}\} C_1 n^{s+\frac{q}{\alpha}-1+\delta q}.
\end{align*}
By partitioning on the  event $\{X_i=Y_i, \forall i \}$ and its complement, using Markov's inequality and preceding results we conclude for the case $\alpha>2$:
\begin{align*}
&P \left( \frac{\ln S_q(n,n^s)}{\ln n} > \max\left\{ s+\frac{q}{\alpha}-1, \frac{sq}{2} \right\} +\delta q+ \epsilon \right)\\
&=P \left( S_q(n,n^s) > n^{\max\left\{ s+\frac{q}{\alpha}-1, \frac{sq}{2} \right\} +\delta q+ \epsilon} \right) \\
&\leq P \Bigg( \max \{1,2^{q-1}\}  \frac{1}{\lfloor n^{1-s} \rfloor} \sum_{i=1}^{\lfloor n^{1-s} \rfloor} \left| \sum_{j=1}^{\lfloor n^s \rfloor} \left( X_{\lfloor n^s \rfloor(i-1)+j} -  EY_{\lfloor n^s \rfloor(i-1)+j, n} \right) \right|^q  \\
&\qquad + \max \{1,2^{q-1}\} C_1 n^{s+\frac{q}{\alpha}-1+\delta q} > n^{\max\left\{ s+\frac{q}{\alpha}-1, \frac{sq}{2} \right\} +\delta q+ \epsilon} \Bigg) \\
&\leq P \Bigg(\max \{1,2^{q-1}\}  \frac{1}{\lfloor n^{1-s} \rfloor} \sum_{i=1}^{\lfloor n^{1-s} \rfloor} \left| \sum_{j=1}^{\lfloor n^s \rfloor} Z_{\lfloor n^s \rfloor(i-1)+j} \right|^q + \max \{1,2^{q-1}\} C_1 n^{s+\frac{q}{\alpha}-1+\delta q} > \\
&\qquad n^{\max\left\{ s+\frac{q}{\alpha}-1, \frac{sq}{2} \right\} +\delta q+ \epsilon} \Bigg) + P \left( \max_{i=1,\dots,n} |X_i|  > n^{\frac{1}{\alpha} + \delta} \right)\\
&\leq \frac{ \max \{1,2^{q-1}\} E \xi_1  + \max \{1,2^{q-1}\} C_1 n^{s+\frac{q}{\alpha}-1+\delta q} }{n^{\max\left\{ s+\frac{q}{\alpha}-1, \frac{sq}{2} \right\} +\delta q + \epsilon}} + C_9 \frac{L(n^{\frac{1}{\alpha}+\delta})}{n^{\alpha \delta}}\\
&\leq \frac{ C_{10} n^{\max\left\{ s+\frac{q}{\alpha}-1+\delta q, \frac{sq}{2} \right\}} }{n^{\max\left\{ s+\frac{q}{\alpha}-1, \frac{sq}{2} \right\} + \delta q + \epsilon}} + C_9 \frac{L(n^{\frac{1}{\alpha}+\delta})}{n^{\alpha \delta}} \to 0.
\end{align*}
Since $\epsilon$ and $\delta$ are arbitrary, it follows
$$\plim_{n \to \infty} \frac{\ln S_q(n,n^s)}{\ln n} \leq \max\left\{ s+\frac{q}{\alpha}-1, \frac{sq}{2} \right\}.$$ In case $\alpha\leq 2$ we can repeat the previous with $n^{s+\frac{q}{\alpha}-1+\delta q}$ instead of
$n^{\max\left\{ s+\frac{q}{\alpha}-1+\delta q, \frac{sq}{2} \right\}}$ and get
$$\plim_{n \to \infty} \frac{\ln S_q(n,n^s)}{\ln n} \leq s+\frac{q}{\alpha}-1.$$
\bigskip

We next show the lower bound in two parts.\\
We first consider the case $\alpha>2$ and assume that $s+\frac{q}{\alpha}-1\leq \frac{sq}{2}$. Denote
$$\sigma^2=\lim_{n\to \infty} \frac{E\left(\sum_{j=1}^n X_j\right)^2}{n},$$
$$\rho_n=P \left( \left| \sum_{j=1}^{\lfloor n^s \rfloor} X_{n^s(i-1)+j} \right| > n^{\frac{s}{2}} \sigma  \right).$$
Since the sequence $X_j$ is stationary and strong mixing with an exponential decaying rate and since $E|X_j|^{2+\zeta} < \infty$ for $\zeta>0$ sufficiently small, the Central Limit Theorem holds (see \cite{hallheyde1980} Corollary 5.1.) and $\sigma^2$ exists. Since $P(|\mathcal{N}(0,1)|>1)>1/4$, it follows that for $n$ large enough $\rho_n > 1/4$. Recall that if $\mathcal{M} \mathcal{B}(n,p)$ is the sum of $n$ stationary mixing indicator variables with expectation $p$ then ergodic theorem implies $\mathcal{M} \mathcal{B}(n,p)/n \to p,$ a.s.
\begin{align*}
&P \left( \frac{\ln S_q(n,n^s)}{\ln n} < \frac{sq}{2} - \epsilon \right) = P \left( S_q(n,n^s) < n^{ \frac{sq}{2} - \epsilon} \right) \\
&\leq P \left( \sum_{i=1}^{\lfloor n^{1-s} \rfloor} \left| \sum_{j=1}^{\lfloor n^s \rfloor} X_{n^s(i-1)+j} \right|^q < n^{ \frac{sq}{2} - \epsilon +1-s} \right) \\
&\leq P \left( \sum_{i=1}^{\lfloor n^{1-s} \rfloor} \mathbf{1} \left( \left| \sum_{j=1}^{\lfloor n^s \rfloor} X_{n^s(i-1)+j} \right| > n^{\frac{s}{2}} \sigma \right) < \frac{n^{ \frac{sq}{2} - \epsilon +1-s}}{n^{\frac{sq}{2}} \sigma^q} \right) \\
&= P \left( \sum_{i=1}^{\lfloor n^{1-s} \rfloor} \mathbf{1} \left( \left| \sum_{j=1}^{\lfloor n^s \rfloor} X_{n^s(i-1)+j} \right| > n^{\frac{s}{2}} \sigma \right) < \frac{n^{1-s-\epsilon }}{\sigma^q} \right)\\
&\leq P \left( \mathcal{M} \mathcal{B}(\lfloor n^{1-s} \rfloor, 1/4)  < \frac{n^{1-s-\epsilon }}{\sigma^q} \right) \to 0,
\end{align*}
hence
$$\plim_{n \to \infty} \frac{\ln S_q(n,n^s)}{\ln n} \geq \frac{sq}{2}.$$
For the second part, assume that $s+\frac{q}{\alpha}-1 > \frac{sq}{2}$. Notice that in this case it must hold $\frac{1}{\alpha}-\frac{s}{2}>0$. We can assume that $\epsilon < \frac{1}{\alpha}-\frac{s}{2}$. Indeed, otherwise we can choose $0<\tilde{\epsilon}<\frac{1}{\alpha}-\frac{s}{2}$ and continue the proof with it in place of $\epsilon$ by observing that
$$P \left( \frac{\ln S_q(n,n^s)}{\ln n} < s+\frac{q}{\alpha}-1 - \epsilon \right) \leq P \left( \frac{\ln S_q(n,n^s)}{\ln n} < s+\frac{q}{\alpha}-1 - \tilde{\epsilon} \right).$$
The main fact behind the following part of the proof is that $\sum |X_i|^q \approx \max |X_i|^q$ and that $s$ is small, which makes the blocks to grow slow. It is generally known that under the assumed mixing condition the asymptotic behaviour of partial maxima is the same as that of the associated independent sequence (see \cite{embrechts1997modelling}). This means that $\max_{j=1,\dots,n} |X_j| / n^{1/\alpha}$ converges in distribution to some positive random variable, so that
$$P \left(\max_{j=1,\dots,n} |X_j| < 2n^{\frac{1}{\alpha}-\epsilon} \right) \to 0.$$
Let $l\in \mathbb{N}$ be such that $|X_l|=\max_{j=1,\dots,n} |X_j|$. Then, for some $k \in \{1,2,\dots,\lfloor n^{1-s} \rfloor \}$ we have $l \in \mathcal{J} :=\{ \lfloor n^s \rfloor (k-1) +1,\dots, \lfloor n^s \rfloor k \}$. Assumption $\alpha>2$ ensures that $E|X_1|^{2+\zeta}<\infty$ for some $\zeta>0$. Applying Markov's inequality and then Lemma \ref{lemma:Rosenthal} yields
\begin{align*}
&P \left( \left| \sum_{j\in \mathcal{J}, j\neq l} X_j \right| > n^{\frac{1}{\alpha}-\epsilon} \right) \leq \frac{ E \left( \sum_{j\in \mathcal{J}, j\neq l} X_j  \right)^2 }{ n^{\frac{2}{\alpha}-2\epsilon}} \\
&\leq \frac{K_1 \sum_{j\in \mathcal{J}, j\neq l} \left( E|X_j|^{2+\zeta} \right)^{\frac{2}{2+\zeta}}}{ n^{\frac{2}{\alpha}-2\epsilon}} \leq \frac{K_2 n^{s}}{ n^{\frac{2}{\alpha}-2\epsilon}} = K_2 n^{s-\frac{2}{\alpha}+2\epsilon} \to 0, \quad \text{as } n\to \infty,
\end{align*}
since $s-\frac{2}{\alpha}+2\epsilon<0$ by the assumption in the proof.
Combining this it follows
\begin{align*}
&P \left( \frac{\ln S_q(n,n^s)}{\ln n} < s+\frac{q}{\alpha}-1 - \epsilon \right) = P \left( S_q(n,n^s) < n^{ s+\frac{q}{\alpha}-1 - \epsilon} \right)\\
&\leq P \left( \sum_{i=1}^{\lfloor n^{1-s} \rfloor} \left| \sum_{j=1}^{\lfloor n^s \rfloor} X_{n^s(i-1)+j} \right|^q < n^{ \frac{q}{\alpha} -q \epsilon} \right) \\
& \leq P \left( \left| \sum_{j\in \mathcal{J}} X_{j} \right|^q < n^{ \frac{q}{\alpha} -q \epsilon} \right) = P \left( \left| \sum_{j\in \mathcal{J}} X_{j} \right| < n^{ \frac{1}{\alpha} -\epsilon} \right) \\
& \leq P \left(|X_l| < 2n^{\frac{1}{\alpha}-\epsilon} \right) +  P \left( \left| \sum_{j\in \mathcal{J}, j\neq l} X_j \right| > n^{\frac{1}{\alpha}-\epsilon} \right) \to 0,
\end{align*}
as $n \to \infty$. Hence,
$$\plim_{n \to \infty} \frac{\ln S_q(n,n^s)}{\ln n} \geq \max\left\{ s+\frac{q}{\alpha}-1, \frac{sq}{2} \right\}.$$
For the case $0<\alpha \leq 2$ we just need a different estimate for the partial sum containing maximum. Choose $\gamma$ such that $0<\gamma<\alpha$. Again we use Markov's inequality
$$P \left( \left| \sum_{j\in \mathcal{J}, j\neq l} X_j \right| > n^{\frac{1}{\alpha}-\epsilon} \right) \leq \frac{ E \left|  \sum_{j\in \mathcal{J}, j\neq l} X_j \right|^{\alpha-\gamma} }{ n^{1-\alpha \epsilon - \frac{\gamma}{\alpha}+\epsilon \gamma}}.$$
From Lemma \ref{lemma:Rosenthal} one can easily bound this expectation by $K_3 n^s$ for some constant $K_3$. Choosing $\epsilon$ and $\gamma$ small enough to make $s-1+\alpha \epsilon + \frac{\gamma}{\alpha}-\epsilon \gamma<0$, we get
$$P \left( \left| \sum_{j\in \mathcal{J}, j\neq l} X_j \right| > n^{\frac{1}{\alpha}-\epsilon} \right) \leq \frac{ K_3 n^s}{ n^{1-\alpha \epsilon - \frac{\gamma}{\alpha}+\epsilon \gamma}} \to 0,  \quad \text{as } n\to \infty,$$
and this completes the (a) part of the proof.\\

\bigskip

(b) Now let $q < \alpha$.\\
We first show the upper bound on the limit, i.e. we analyse
\begin{align*}
&P \left( \frac{\ln S_q(n,n^s)}{\ln n} > \frac{sq}{\beta(\alpha)} +\epsilon \right)
=P \left( S_q(n,n^s) > n^{\frac{sq}{\beta(\alpha)} +\epsilon} \right) \\
&\leq P \left(\frac{1}{\lfloor n^{s-1} \rfloor} \sum_{i=1}^{\lfloor n^{1-s} \rfloor} \left| \sum_{j=1}^{\lfloor n^s \rfloor} X_{n^s(i-1)+j,n} \right|^q > n^{\frac{sq}{\beta(\alpha)} +\epsilon} \right)\leq \frac{  E \left| \sum_{j=1}^{\lfloor n^s \rfloor} X_{j} \right|^q }{n^{\frac{sq}{\beta(\alpha)} +\epsilon}} ,
\end{align*}
where we write $\beta(\alpha)=\alpha$ or $2$ corresponding to $\alpha\leq 2$ or $\alpha>2$.
To show that this tends to zero, we first consider the case $\alpha>2$. If $q>2$, using Lemma \ref{lemma:Rosenthal} with $\zeta$ small enough it follows
$$E \left| \sum_{j=1}^{\lfloor n^s \rfloor} X_{j} \right|^q \leq C_1 \max \{n^s, n^{\frac{sq}{2}} \}.$$
For the case $q\leq 2$ we combine Jensen's inequality with Lemma \ref{lemma:Rosenthal}
$$E \left| \sum_{j=1}^{\lfloor n^s \rfloor} X_{j} \right|^q \leq \left( E \left| \sum_{j=1}^{\lfloor n^s \rfloor} X_{j} \right|^2 \right)^{\frac{q}{2}} \leq C_2 n^{\frac{sq}{2}}.$$
In case $\alpha \leq 2$ we choose $\gamma$ small enough to make $q<\alpha-\gamma < \alpha$ and we get
$$E \left| \sum_{j=1}^{\lfloor n^s \rfloor} X_{j} \right|^q \leq \left( E \left| \sum_{j=1}^{\lfloor n^s \rfloor} X_{j} \right|^{\alpha-\gamma} \right)^{\frac{q}{\alpha - \gamma}} \leq C_3 n^{\frac{sq}{\alpha-\gamma}}.$$
We next prove the lower bound. For the case $\alpha>2$ the proof is the same as the proof of (a). Assume $\alpha\leq 2$. The arguments go along the same line, but here we avoid using limit theorems for partial sums of stationary sequences. Instead we use before mentioned asymptotic behaviour of the partial maximum, that is, we use the fact that $\max_{j=1,\dots,\lfloor n^s \rfloor} |X_j| / n^{s/\alpha}$ converges in distribution to some positive random variable. This means we can choose some constant $m>0$ such that for large enough $n$
$$P \left( \frac{\max_{j=1,\dots,\lfloor n^s \rfloor} |X_{j}| }{n^{\frac{s}{\alpha}}} > 2m \right) > \frac{1}{4}.$$
Denote $|X_l|=\max_{j=1,\dots,\lfloor n^s \rfloor} |X_j|$. Then it follows that
$$P \left( \left| \sum_{j=1}^{\lfloor n^s \rfloor} X_{j} \right| > m n^{\frac{s}{\alpha}}  \right) \geq P \left( |X_l|> 2 m n^{\frac{s}{\alpha}} \right) + P \left( \left| \sum_{j=1, j \neq l}^{\lfloor n^s \rfloor} X_{j} \right| < m n^{\frac{s}{\alpha}} \right) > \frac{1}{4}.$$
Now we conclude as before, denoting by $\mathcal{M} \mathcal{B}(n,p)$ the sum of $n$ stationary mixing indicator variables with mean $p$ and noting that ergodic theorem implies $\mathcal{M} \mathcal{B}(n,p)/n \to p>0,$ a.s.
\begin{align*}
&P \left( \frac{\ln S_q(n,n^s)}{\ln n} < \frac{sq}{\alpha} - \epsilon \right) = P \left( S_q(n,n^s) < n^{ \frac{sq}{\alpha} - \epsilon} \right) \\
&\leq P \left( \sum_{i=1}^{\lfloor n^{1-s} \rfloor} \left| \sum_{j=1}^{\lfloor n^s \rfloor} X_{n^s(i-1)+j} \right|^q < n^{ \frac{sq}{\alpha} - \epsilon +1-s} \right) \\
&\leq P \left( \sum_{i=1}^{\lfloor n^{1-s} \rfloor} \mathbf{1} \left( \left| \sum_{j=1}^{\lfloor n^s \rfloor} X_{n^s(i-1)+j} \right| > n^{\frac{s}{\alpha}} m \right) < \frac{n^{ \frac{sq}{\alpha} - \epsilon +1-s}}{n^{\frac{sq}{\alpha}} m^q} \right) \\
&\leq P \left( \sum_{i=1}^{\lfloor n^{1-s} \rfloor} \mathbf{1} \left( \left| \sum_{j=1}^{\lfloor n^s \rfloor} X_{n^s(i-1)+j} \right| > n^{\frac{s}{\alpha}} m \right) < \frac{n^{ 1-s- \epsilon}}{ m^q} \right)\\
&\leq P \left( \mathcal{M} \mathcal{B}(\lfloor n^{1-s} \rfloor, 1/4)  <  \frac{n^{ 1-s- \epsilon}}{ m^q} \right) \to 0,
\end{align*}
and this proves the lower bound.

\bigskip

(c) It remains to consider the case $q = \alpha$. But this is simple since for every $\delta>0$
$$\frac{\ln S_{q-\delta}(n,n^s)}{\ln n} \leq \frac{\ln S_q(n,n^s)}{\ln n} \leq \frac{\ln S_{q+\delta}(n,n^s)}{\ln n},$$
so the limit must be continuous in $q$ and the claim follows from previous cases.

\end{proof}

\begin{proof}[Proof of theorem \ref{thm:tauhatasymptotic}]
Fix $q>0$ and denote $y_n (s) =\ln S_q(n,n^s) / \ln n$. We first show that
\begin{equation}\label{plim}
\plim_{n \to \infty} \hat{\tau}_{N,n} (q) = \dfrac{\sum_{i=1}^{N-1} \frac{i}{N} R_{\alpha}(q,\frac{i}{N}) - \frac{1}{N-1} \sum_{i=1}^{N-1} \frac{i}{N} \sum_{j=1}^{N-1} R_{\alpha}(q,\frac{i}{N})}{ \sum_{i=1}^{N-1} \left(\frac{i}{N}\right)^2 - \frac{1}{N-1} \left( \sum_{i=1}^{N-1} \frac{i}{N} \right)^2 }.
\end{equation}
Let $\varepsilon>0$ and $\delta>0$. By Theorem \ref{thm:main}, for each $i=1,\dots,N$ there exists $n_i$ such that
$$P \left( \left| y_n(\frac{i}{N}) - R_{\alpha}(q,\frac{i}{N}) \right| > \frac{\varepsilon}{N-1} \right) < \frac{\delta}{N-1},$$
for $n \geq n_i$. Take $n_{max}= \max_{i=1,\dots,N} n_i$. Then for all $n \geq n_{max}$,
\begin{align*}
P &\left( \sum_{i=1}^{N-1} \frac{i}{N}  \left| y_n(\frac{i}{N}) - R_{\alpha}(q,\frac{i}{N}) \right| > \varepsilon \right) \leq P \left( \sum_{i=1}^{N-1} \left| y_n(\frac{i}{N}) - R_{\alpha}(q,\frac{i}{N}) \right| > \varepsilon \right) \\
&\leq (N-1) P \left( \left| y_n(\frac{i}{N}) - R_{\alpha}(q,\frac{i}{N}) \right| > \frac{\varepsilon}{N-1} \right) < \delta.
\end{align*}
This proves the convergence for two terms depending on $n$ and claim now follows by continuous mapping theorem. By dividing denominator and numerator of the fraction in limit \eqref{plim} by $1/(N-1)$, one can see all the sums involved as Riemann sums based on equidistant partition. Functions involved, $s \mapsto s R_{\alpha}(q,s)$, $s \mapsto R_{\alpha}(q,s)$, $s \mapsto s $ and $s \mapsto s^2$, are all bounded continuous on $[0,1]$, so all sums converge to integrals when partition is refined, i.e. when $N \to \infty$. Thus
\begin{equation*}
\lim_{N \to \infty} \plim_{n \to \infty} \hat{\tau}_{N,n} (q) = \dfrac{\int_{0}^{1} s R_{\alpha}(q,s) ds - \int_0^1 s ds \int_0^1 R_{\alpha}(q,s) ds }{ \int_0^1 s^2 ds - \left( \int_0^1 s ds \right)^2 }.
\end{equation*}
Solving the integrals using expression for $R_{\alpha}(q,s)$, one gets $\tau(q)$ as in \eqref{tau}.
\end{proof}

\bibliographystyle{chicago}
\bibliography{References}

\end{document}